\documentclass[runningheads]{llncs}
\usepackage{amssymb}
\usepackage{amsmath}
\usepackage{epstopdf}
\usepackage{shuffle}
\usepackage{tikz}
\usepackage{ulem}
\usepackage{graphicx}
\usepackage[hyphens]{url}
\usepackage{microtype}
\usepackage{url}
\def\P{{\mathchoice {\hbox{$\sf\textstyle P\kern-0.4em Z$}}
{\hbox{$\sf\textstyle P\kern-0.4em P$}}
{\hbox{$\sf\scriptstyle P\kern-0.3em P$}}
{\hbox{$\sf\scriptscriptstyle P\kern-0.2em P$}}}}

\usepackage{epsf}
\usepackage{enumerate}
\usepackage{enumitem}
\def\P{{\mathchoice {\hbox{$\sf\textstyle P\kern-0.4em Z$}}
{\hbx{$\sf\textstyle P\kern-0.4em P$}}
{\hbox{$\sf\scriptstyle P\kern-0.3em P$}}
{\hbox{$\sf\scriptscriptstyle P\kern-0.2em P$}}}}

\let\Rightarrow=\Rightarrow

\newtheorem{fait}{Claim}

\def\ex#1\par{\par\noindent\begin{exemple} \nopagebreak \strut \rm #1 \end{exemple}}
\def\thm#1\par{\medskip\par\noindent\begin{theorem} \strut \sl #1 \end{theorem}\par}
\def\propo#1\par{\medskip\par\noindent\begin{proposition} \strut \sl #1 \end{proposition} \par}
\def\proo#1\par{\medskip\par\noindent{\it Proof.} \strut \rm #1 $\Box$ \par}
\def\prob#1\par{\medskip\par\noindent\begin{prob} \strut \sl #1 \end{prob} \par}
\def\cor#1\par{\medskip\par\noindent\begin{corollary} \strut \sl #1 \end{corollary}\par}
\def\lm#1\par{\medskip\par\noindent\begin{lemma} \strut \sl #1 \end{lemma}\par}
\def\defil#1\par{\medskip\par\noindent\begin{condit} \strut \sl #1\end{condit}\par}
\def\fct#1\par{\par\noindent\begin{fait}  \nopagebreak \strut #1 \end{fait}}
\def\fct#1\par{\par\noindent\begin{claim} \nopagebreak \strut #1 \end{claim}}
\def\defi#1\par{\medskip\par\noindent{\begin{defin} \strut  \sl #1 \end{defin}}\par}
\def\nota#1\par{\par\noindent\begin{notat} \nopagebreak  \strut #1 \end{notat}}

\def\rem#1\par{\par\noindent\begin{rema} \nopagebreak \strut \rm #1 \end{rema}}

\date{}

\begin{document}

\thispagestyle{empty}
\title{When Variable-Length Codes Meet  the Field of Error Detection}
\titlerunning{When Variable-Length Codes Meet the Field of Error Detection}
\author{Jean N\'eraud \footnote{Corresponding author}}
\institute{Universit\'e de Rouen, Laboratoire d'Informatique, de Traitement de l'Information et des Syst\`emes, 
 Avenue de l'Universit\'e, 76800 Saint-\'Etienne-du-Rouvray, France\\
\email{neraud.jean@gmail.com}\\
\url{http:neraud.jean.free.fr; orcid: 0000-0002-9630-461X} 
}
\toctitle{v}
\tocauthor{J.~N\'eraud}
\authorrunning{J. N\'eraud}
\maketitle
\setcounter{footnote}{0}
\begin{abstract}
Given  a finite alphabet $A$ and a  binary relation $\tau\subseteq A^*\times A^*$, a  set $X$
is $\tau$-{\it independent}
 if $ \tau(X)\cap X=\emptyset$. 
Given a quasi-metric $d$ over $A^*$ (in the meaning of \cite{W31}) and   $k\ge 1$,  we associate the relation $\tau_{d,k}$
defined by $(x,y)\in\tau_{d,k}$ if, and only if, $d(x,y)\le k$ \cite{CP02}.
In the spirit of \cite{JK97,N21}, the error detection-correction  capability of variable-length codes  can be expressed in term of  conditions over $\tau_{d,k}$.
With respect to  the prefix metric, the  factor one, and every quasi-metric  associated to (anti-)automorphisms of the free monoid, we examine whether those conditions are decidable for a given regular code.
\keywords{
 Anti-reflexive, automaton, automorphism, anti-automorphism, Bernoulli measure, binary relation, channel,
code, codeword, complete, distance,
error correction, error detection,
embedding, factor, free monoid, homomorphism, independent,  input word, Kraft inequality,
maximal, measure, 
metric, monoid,  output word, prefix, quasi-metric, regular, subsequences,
suffix, synchronization constraint, transducer, variable-length code,
word}
\end{abstract}
\section{Introduction}
In Computer Science, the transmission of finite sequences of symbols (the so-called {\it words}) via some  channel
constitutes one of the most challenging research fields.
With the notation of the free monoid theory, some classical models 
may be informally described as indicated in the following:

Two finite  {\it alphabets}, say  $A$ and $B$,  are required,
every information being modeled by a unique word, say  $u$,  in $B^*$ (the {\it free monoid} generated by $B$). 
Usually, in order to  facilitate the transmission, beforehand $u$ is transformed in  $w\in A^*$, the so-called {\it input word}: this is  done by applying some fixed one-to-one {\it coding} mapping  $\phi: B^*\longrightarrow A^*$. 
In numerous cases, $\phi$ is a {\it monoid homomorphism},
whence $X=\phi(B)$ is a  {\it  variable-length code} (for short, a {\it code}): equivalently 
every equation among the words of $X$ is necessarily trivial. Such a translation  is particularly illustrated by the well-known examples of the Morse and Huffman codes. 
Next, $w$ is transmitted via a fixed {\it channel} into $w'\in A^*$, the so-called {\it output word}: should $w'$ be altered by some
{\it noise}
and then the resulting word $\phi^{-1}(w')\in B^*$ could be different from the initial word $u$.
In the most general model of  transmission, the channel is represented by some {\it probabilistic transducer}.
However, in the framework of error detection, most of the models only require that highly likely errors need to be taken into account:
in the present paper, 
we assume the transmission channel modeled by some {\it binary word relation}, namely $\tau\subseteq A^*\times A^*$.
In order to retrieve $u$, the homomorphism $\phi$, and thus the code $X$, must satisfy specific
constraints, which of course depend of the channel $\tau$:
in view of some  formalization, we denote by $\widehat\tau$ the {\it reflexive closure} of $\tau$, and by  $\underline\tau$ its  {\it anti-reflexive restriction} that is, 
$\tau\setminus\{(w,w)|w\in A^*\}$.

About the channel itself, 
the so-called 
{\it synchronization constraint} appears mandatory: it states that,
for each input word  factorized $w=x_1\cdots x_n$, where $x_1,\cdots,x_n$ are {\it codewords} in $X$,
every output word has to be factorized 
$w'= x'_1\cdots x'_n$, with $(x_1,x'_1),\cdots, (x_n,x'_n)\in\widehat\tau$. 
In order to ensure such a constraint,  as for the Morse code, some pause symbol could be inserted after each codeword $x_i$. 

With regard to the code $X$,  in order to minimize the number of errors, in most cases 
some close neighbourhood constraint over  $\widehat\tau(X)$  is applied. 
In the most frequent use, such a constraint consists of some minimal distance condition:
the smaller the distance between the input codeword $x\in X$ and any of its  corresponding output words $x'\in \widehat\tau(X)$,
the more optimal is error detection. 
In view of that, we fix over $A^*$ a {\it quasi-metric} $d$, in the meaning of \cite{W31} (the difference with a {\it metric} is that $d$ needs not to satisfy the symmetry axiom).
As outlined in \cite{CP02}, given an error tolerance level  $k\ge 0$, a corresponding  binary word relation, that we denote by  $\tau_{d,k}$,  can be associated in such a way that 
$(w,w')\in\tau_{d,k}$ is equivalent to $d(w,w')\le k$.
Below, in the spirit of \cite{JK97,N21}, we draw some specification regarding  error detection-correction capability.
Recall that a subset  $X$ of $A^*$ is {\it independent} with respect to  $\tau\subseteq A^*\times A^*$ (for short, $\tau$-{\it independent}) whenever $\tau(X)\cap X=\emptyset$: this notion, which appears dual with the one  of  {\it closed code} \cite{N21},  relies  to the famous {\it dependence systems} \cite{C81,JK97}.
Given a family of  codes, say ${\cal F}$, a code  $X\in\cal F$  is {\it maximal in $\cal F$}
whenever $X\subseteq Y$, with $Y\in {\cal F}$, implies $Y=X$.
We introduce the four following conditions:
\begin{enumerate} [label={\rm (c\arabic*)},wide=10pt]
\item \label{1} 
{\it Error detection:}
 $X$ is $\uline{\tau_{d,k}}$-{\it independent}. 
\item
\itemsep2pt
\label{2}
{\it Error correction:}
$x,y\in X$ and $\tau_{d,k}(x)\cap\tau_{d,k}(y)\neq\emptyset$ implies  $x=y$.
\item \label{3}   {\it $X$ is maximal in the family of \uline{$\tau_{d,k}$}-independent  codes}.
\item \label{4}  {\it $\widehat{\tau_{d,k}}(X)$ is a code.}
\end{enumerate}
A few  comments on Conds. \ref{1}--\ref{4}:

-- By definition, Cond. \ref{1} is satisfied 
 if, and only if, the distance between different elements of $X$ is greater than $k$
that is,  the code $X$ can detect at most $k$ errors in the transmission of every codeword $x\in X$.

-- Cond. \ref{2}  states a classical definition: 
equivalently, for every codeword $x$ we have $\tau_{d,k}^{-1}\left(\tau_{d,k}(x)\right)\cap X=\{x\}$ whenever $ \tau_{d,k}(x)$ is non-empty \cite{N21}.  

-- With Cond. \ref{3}, in the family of \uline{$\tau_{d,k}$}-independent codes, $X$ cannot be improved. From this point of view,  fruitful investigations have been done in several classes determined by code properties \cite{JKK01,KKK14,L00,L01}.
On the other hand, according to the famous Kraft inequality, given a {\it positive Bernoulli measure} over $A^*$,  say $\mu$, for every
 code $X$ we have $\mu(X)\le 1$.
According to a famous result due to Sch\"utzenberger, given a {\it regular} code $X$, the condition $\mu(X) =1$ itself corresponds to $X$ being maximal in the whole family of codes, or equivalently $X$ being  {\it complete} that is, every word in $A^*$ is a {\it factor}
of some word in $X^*$. From this last point of view, no part of $X^*$ appears spoiled.

-- At last, Cond. \ref{4} 
expresses that the factorization of every output
message over the set  $\widehat{\tau_{d,k}}(X)$ is done in a unique way.  Since $d$ is a quasimetric, the corresponding relation ${\tau_{d,k}}$ is reflexive, therefore Cond. \ref{4} is equivalent to ${\tau_{d,k}}(X)$ being a code.

\smallskip
As shown in \cite{JK97,N21}, there are regular codes satisfying \ref{1} that cannot satisfy  \ref{2}.
Actually, in most of the cases it could be very difficult, even impossible, to satisfy all together Conds. \ref{1}--\ref{4}: 
necessarily some compromise has to be adopted.
In view of this, given a regular code $X$, a natural question consists in examining whether each of  those conditions is satisfied  in the framework of some special quasi-metric.
From this point of view, in \cite{N21}, we considered  the so-called  {\it edit relations},
 some peculiar  compositions of one-character {\it deletion}, {\it insertion}, and {\it substitution}:
such relations involve  the famous Levenshtein and Hamming metrics \cite{H50,K83,L65}, which  are prioritary related to subsequences in words.
In the present paper, we  focuse on quasimetrics  rather involving factors:

-- The {\it prefix} metric is defined  by $d_{\rm P}(w,w')=|w|+|w'|-2|w\wedge w'|$,
where  $w\wedge w'$ stands for the longest common {\it prefix} of $w$ and $w'$:
we set ${\cal P}_k=\tau_{d_{\rm P},k}$.

-- The {\it factor} metric itself is
defined by $d_{\rm F}(w,w')=|w|+|w'|-2|f|$, where $f$ is a maximum length  common factor of $w$, $w'$: we set ${\cal F}_k=\tau_{d_{\rm F},k}$.

-- A third type of topology can be introduced in connection with {\it monoid automorphisms} or {\it anti-automorphisms} (for short, we write {\it (anti-)automorphisms}):
such a topology particularly concerns the domain of  DNA sequence comparison. 
By anti-automorphism of the free monoid, we mean a one-to-one mapping onto $A^*$, say $\theta$, 
such that  the equation  $\theta(uv)=\theta(v)\theta(u)$ holds for  every $u,v\in A^*$
(for involvements of such mappings in the dual notion of closed code, see \cite{NS20}).
With every  (anti-)automorphism $\theta$ we associate  the quasi-metric $d_\theta$, defined as follows:

\smallskip
(1) $d_\theta(w,w')=0$ is equivalent to $w=w'$;

(2) we have $d_\theta(w,w')=1$ whenever $w'=\theta(w)$ holds, with $w\ne w'$;

(3)  in all other cases we set $d_\theta(w,w')=2$. 

\smallskip 
\noindent
  By definition we have $\tau_{{d_\theta},1}=\widehat{\theta}$ and $\underline{\tau_{_\theta,1}}=\underline\theta$.
In addition, a code $X$ is $\underline\theta$-independent if, and only if, for every pair of different words $x,y\in X$, we have $d_\theta(x,y)=2$ that is, $X$ is capable to detect at most one error.
{\flushleft We will establish  the following result:}
{\flushleft
{\bf Theorem.}}
{\it With the preceding notation, given a regular code $X\subseteq A^*$,
it can be decided whether $X$ satisfies each of the following conditions:

{\rm (i)}  Conds.  \ref{1}--\ref{4} wrt.   ${\cal P}_k$.

{\rm (ii)}   Conds.  \ref{3}, \ref{4} wrt.  ${\cal F}_k$.

{\rm (iii)} If   $X$ is  finite, Conds. \ref{1}, \ref{2} wrt.  ${\cal F}_k$.

{\rm (iv)}  Conds. \ref{3}, \ref{4} wrt. $\widehat \theta$, for any (anti-)automorphism  $\theta$  of $A^*$.
In that framework,  in any case $X$ satisfies Conds. (c1), (c2), which  are actually equivalent.}\\


\noindent  Some comments regarding the proof:

-- For proving that $X$ satisfies Cond. \ref{1} wrt. ${\cal P}_k$, the main argument consists in establishing that  $\underline{{\cal P}_k}$
 is a  {\it regular} relation \cite[Ch. IV]{S03} that is, $\underline{{\cal P}_k}$   can be simulated by a finite transducer.

-- Once more wrt. ${\cal P}_k$, we actually prove that  Cond. \ref{2} is equivalent to $(X\times X)\cap \underline{{\cal P}_{2k}}=\emptyset$. 
Since $X\times X$ is  a {\it recognizable} relation \cite [Ch. IV]{S03}, it can be decided whether that equation holds.

-- Regarding Cond. \ref{3}, the critical step is reached by proving that, wrt. each of the quasi-metrics raised in the paper, for a regular code $X$, being  maximal in the family of $\underline{\tau_{d,k}}$-independent codes is equivalent to being complete that is, $\mu(X)=1$.
This is done by proving that every non-complete $\underline{\tau_{d,k}}$-independent code, say $X$,  can be embedded into some complete $\underline{\tau_{d,k}}$-independent one: 
in other words, $X$ cannot be maximal in such a  family of codes.
 In order to establish such a property, 
in the spirit of  \cite{BWZ90,L03,N06,N08,NS20,ZS95}, we provide specific regularity-preserving embedding formulas,
whose schemes are based upon the method from \cite{ER85}.
Notice that,  in   \cite{JKK01,KM15,L00,L01,VVH05}, wrt.  peculiar families of sets, algorithmic methods for embedding a set into  some maximal one were also provided.

-- With regard to Cond. \ref{4}, for each of the preceding relations, the set $\widehat{\tau_{d,k}}(X)={\tau_{d,k}}(X)$ is regular therefore, in any case, 
by applying the famous Sardinas and Patterson algorithm \cite{SP53}, one can decide whether that condition is satisfied.

\medbreak
\noindent We now shorty describe the contents of the paper:

-- Section \ref{prelim} is devoted to the preliminaries: we recall fundamental notions over codes,  regular (resp., {\it recognizable}) relations, and automata. 

-- The aim of Sect. \ref{Prefix-metric} is to study ${\cal P}_k$. We prove that, in any case, the corresponding  relation $\underline{{\cal P}_k}$ is itself regular.
Furthermore, given a regular code $X$, one can decide whether $X$ satisfies any of   Conds. \ref{1}--\ref{4}. Some remarks are also formulated regarding the so-called  {\it suffix} metric.
 
-- Sect. \ref{Phi} is concerned with the factor metric. We prove that, given a finite code, one can decide whether it satisfies any of  Conds. \ref{1}--\ref{4}.
For a non-finite regular codes, we prove that one can decide whether it satisfies Conds. \ref{3}, \ref{4}, however,  the question of the decidability of Conds. \ref{1}, \ref{2} remains open.

-- Sect. \ref{antiautomorphism} is devoted to 
(anti-)automorphisms. We obtain results similar to those involving the  relation ${\cal P}_k$.
 
-- In Sect. \ref{conclusion}, the paper concludes  with some possible directions for further research. 
\section{Preliminaries}
\label{prelim}
Several definitions and notation has already been stated in the introduction. In the whole paper, we fix a finite alphabet $A$, with $|A|\ge 2$, and we denote by $\varepsilon$ 
the word of length $0$.
Given two words $v,w\in A^*$, $v$ is a {\it prefix} (resp., {\it suffix}, {\it factor}) of $w$ if 
words $u,u'$ exist such that $w=vu$ (resp., $w=u'v$, $w=u'vu$). 
We denote by ${\rm P}(w)$ (resp.,  ${\rm S}(w)$, ${\rm F}(w)$) the set of the words that are prefix (resp., suffix,  factors) of $w$.
More generally,  given a set $X\subseteq A^*$, we set  ${\rm P}(X)=\bigcup_{w\in X}{\rm P}(w)$; the sets  ${\rm S}(X)$ and ${\rm F}(X)$ are defined in a similar way.
A word $w\in A^*$, is {\it overlapping-free} if $wv\in A^*w$, with $|v|\le |w|-1$,   implies $v=\varepsilon$.

{\flushleft {\it Variable-length codes}}

We assume that the reader has a fundamental understanding  with the main concepts of the theory of variable-length codes: we  suggest, if necessary,  that  he (she) refers to \cite{BPR10}.
Given a subset $X$ of $A$, and $w \in X^*$, let $x_1, \cdots, x_n\in X$ such that $w$ is the result of the concatenation of the words $x_1, x_2,\cdots, x_n$, in this order. 
In view of specifying the factorization of $w$ over X, we use the notation $w =(x_1)(x_2)\cdots (x_n)$, or equivalently: $w =x_1\cdot x_2\cdots x_n$. 
For instance, over the set $X=\{a, ab, ba\}$, the word $bab\in X^*$ can be factorized as  $(ba)(b)$ or $(b)(ab)$ (equivalently denoted by $ba\cdot b$ or  $b\cdot ab$).

A set $X$ is a {\it variable-length code} (a {\it code} for short) if  for any pair of finite sequences of words in $X$, say  $(x_i)_{1\le i\le n}$, $(y_j)_{1\le j\le p}$, the equation
$x_1\cdots x_n=y_1\cdots y_p$ implies  $n=p$, and $x_i=y_i$ for each integer $i\in [1,n]$ (equivalently, the submonoid $X^*$ is {\it free}).
In other words, every element of $X^*$ has a unique factorization over $X$.
A set  $X\ne\{\varepsilon\}$ is a {\it prefix} (resp., {\it suffix}) {\it code} if $x\in {\rm P}(y)$ (resp., $x\in {\rm S}(y)$) implies $x=y$,  for every pair of words $x,y\in X$; 
$X$ is a {\it bifix} code  if it is both a prefix code and a suffix one. A set $X\subseteq A^*$ is {\it uniform} if all its elements have a common length.
In the case where we have $X\neq\{ \varepsilon\}$, the uniform set $X$ is a bifix code.

Given a  regular set $X$,  the Sardinas and Patterson  algorithm allows  to decide whether or not $X$ is a  code.
Since it will be used several times through the paper,
it is convenient to shortly recall it. Actually, some ultimately periodic sequence of sets, namely $(U_n)_{n\ge 0}$, is computed, as indicated in the following:
\begin{eqnarray}
\label{SP}
U_0=X^{-1}X\setminus\{\varepsilon\}~~~{\rm and:}~~~
(\forall n\ge 0)~~~U_{n+1}=U_n^{-1}X\cup X^{-1}U_n.
\end{eqnarray}
The algorithm necessarily stops: this corresponds to either  $\varepsilon\in U_n$ or  $U_n=U_p$, for some pair of different integers $p<n$: 
 $X$ is  a code if, and only if,  the second condition holds. 

A positive {\it Bernoulli distribution} consists in some total mapping $\mu$ from the alphabet $A$ into ${\mathbb R}_+=\{x\in {\mathbb R}: x\ge 0\}$ (the set of the non-negative real numbers)  such that   $\sum_{a\in A}\mu(a)=1$;
that mapping  is extended into a unique  monoid homomorphism from  $A^*$ into $({\mathbb R}_+,\times)$, which is itself extended into a unique  positive measure  $\mu:2^{A^*}\longrightarrow {\mathbb R}_+$.
 In order to do this, for each word $w\in A^*$, we set  $\mu\left(\{w\}\right)=\mu(w)$; in addition given two disjoint subsets $X,Y$ of $A^*$, we set $\mu(X\cup Y)=\mu(X)+\mu(Y)$.
In the whole paper, we take for $\mu$  the  so-called {\it uniform} Bernoulli measure: it is determined by $\mu(a)=1/|A|$, for each $a\in A$.

The following results are classical: the first one is due to Sch\"utzenberger and the second provides some answer to a question actually stated in \cite{R77}.
\begin{theorem}{\rm \cite[Theorem 2.5.16]{BPR10}}
\label{classic}
Given a  regular code $X\subseteq A^*$,   the following properties are equivalent:

{\rm (i)} $X$ is complete;

{\rm (ii)} $X$ is a maximal  code;

{\rm (iii)} we have   $\mu(X)=1$.
\end{theorem}
\begin{theorem}{\rm \cite{ER85}}.
\label{EhRz}
Let $X\subseteq A^*$ be non-complete 
code, $z\notin {\rm F}(X^*)$ overlapping-free,  $U=A^*\setminus\left( X^*\cup A^*zA^*\right)$, and $Y=(zU)^*z$. Then $Z=X\cup Y$ is a   
complete code.
\end{theorem}
Clearly, if $X$ is  a regular set then the same holds for the resulting set $Z$.

\bigskip
\noindent
{\it Regular relations, recognizable relations}

We also assume  the reader to be familiar  with the  theory of regular relations and automata:  if necessary,  we suggest that  he (she) refers to \cite{E74} or
\cite[Ch. IV]
{S03}. 

Given two monoids, say  $M,N$, 
a binary relation from $M$ into $N$ consists in any subset $\tau\subseteq M\times N$. 
For $(w,w')\in\tau$, we also set $w'\in\tau(w)$, and we  set $\tau(X)=\{\tau(x): x\in X\}$. 
The composition in this order of  $\tau$ by $\tau'$ is defined by  $\tau\cdot\tau'(x)=
\tau'\left(\tau(x)\right)$ (the notation  $\tau^k$ refers to that operation);
 $\tau^{-1}$, the  {\it inverse} of $\tau$, is defined by  $(w,w')\in\tau^{-1}$ whenever $(w',w)\in\tau$. We denote by $\overline\tau$ the {\it complement} of $\tau$, i.e.  $(M\times N)\setminus \tau$.

A family of subsets of $M$, say ${\cal F}\subseteq 2^M$, is {\it  regularly closed} if, and only if, the sets $X\cup Y$, $XY$, and  $X^*$ belong to ${\cal F}$, whenever we have $X,Y\in {\cal F}$.
Given a family ${\cal F}\subseteq 2^M$, its {\it regular closure} is the smallest (wrt. the inclusion) subset of $2^M$ that contains ${\cal F}$ and which is regularly closed.
A  binary relation $\tau\subseteq M\times N$ is {\it regular} (or equivalently, {\it rational}) 
if, and only if,  it belongs to the regular closure of the finite subsets of $M\times N$.
Equivalently there is some finite  $M\times N$-{\it automaton} (or equivalently, {\it transducer}),
say ${\cal R}$,  with {\it behavior}  $\left|{\cal R}\right|=\tau$
 \cite{EM65,S03}. 
The family of regular relations is closed under  inverse and composition. 

The so-called  {\it recognizable}  relations constitute a noticeable subfamily in regular relations:
a  subset $R\subseteq M\times N$ is {\it  recognizable} if, and only if, 
we have $R=R\cdot\phi\cdot\phi^{-1}$,
 for some monoid homomorphism $\phi:M\times N \longrightarrow P$, with $P$ finite.
Equivalently, finite families of recognizable subsets of $M$ and $N$, namely $(T_i)_{\in I}$ and $(U_i)_{i\in I}$, exist such that $R=\bigcup
_{i\in I} T_i\times U_i$ \cite[Corollary II.2.20]{S03}.

In the paper, we focus on $M=N=A^*$.
With this condition, recognizable relations 
 are closed under composition, complement and intersection,  their intersection with a regular relation being itself  regular
\cite[Sect. IV.1.4]{S03}.
 According to \cite[Theorem IV.1.3]{S03}, given a regular relation $\tau\subseteq A^*\times A^*$,
and a regular (equivalently, recognizable) set $X\subseteq A^*$, the sets
$\tau(X)$ and $\tau^{-1}(X)$ are regular.
If $X$, $Y$ are recognizable  subsets of $A^*$, the same holds for  $X\times Y$.
At last the relation $id_{A^*}=\{(w,w)|w\in A^*\}$ and its complement  $\overline {id_{A^*}}$, are regular but non-recognizable. 
\section{Error detection and the  prefix metric}
\label{Prefix-metric}
Given $w,w'\in A^*$, a unique pair of words $u,u'$ exist such that  $w=(w\wedge w')u$ and $w'=(w\wedge w')u'$: 
by definition, we have $d_{\rm P}(w,w')=|u|+|u'|$ (see Figure 1).
\begin{figure}
\begin{center}
\label{Figure0}
\includegraphics[width=8.5cm,height=4.5cm]{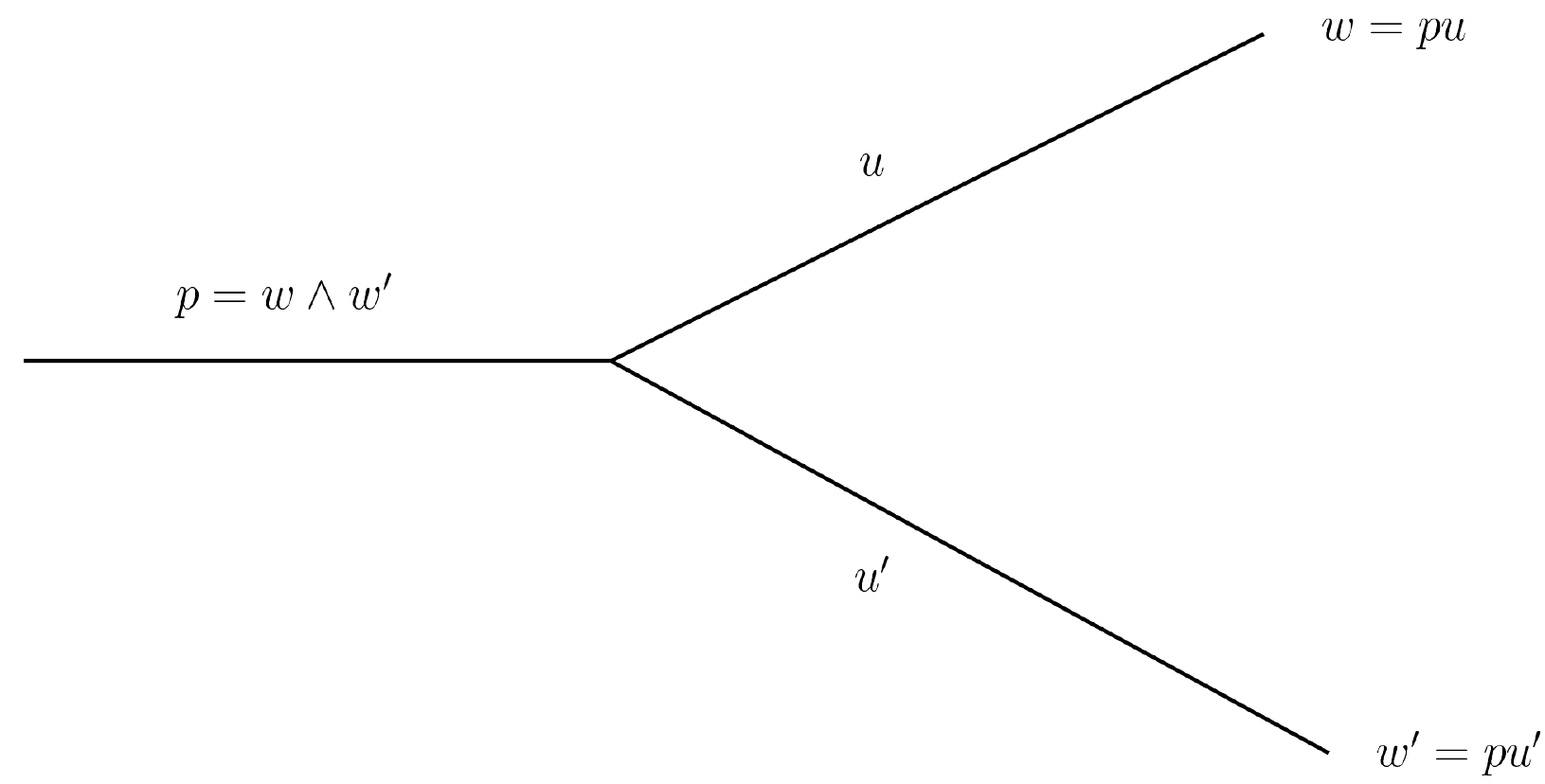}
\end{center}
\caption[]{We have $(w,w')\in{\cal P}_k$ iff.  $|u|+|u'|\le k$; similarly,
$(w,w')\in\underline{{\cal P}_k}$ is equivalent to $1\le |u|+|u'|\le k$.}
\end{figure}
In addition ${\cal P}_k$ is reflexive and symmetric that is, the equality ${\cal P}_k^{-1}={\cal P}_k$ holds, and we have ${\cal P}_{k}\subseteq {\cal P}_{k+1}$.
Below, we provide  some example:
\begin{example}
\label{Prefs}
 Let  $A=\{a,b\}$. 
(1) The finite prefix  code  $X=\{a,ba,b^2\}$ satisfies Cond. \ref{1} wrt.  ${\rm \cal P}_1$ (in other words $X$ is $1$-error-detecting).
Indeed, it follows from $d_{\rm P}(a,ba)=d_{\rm P}(a,bb)=3$ and $d_{\rm P}(ba,bb)=2$ that we have $(x,y)\notin {\cal P}_1$, for each pair of different words $x,y\in X$, that is  $\underline{{\rm \cal P}_1}(X)\cap X=\emptyset$.
 Cond. \ref{2} is not satisfied by $X$. Indeed, we have:

 ${\cal P}_1^{-1}{\cal P}_1\left(ba\right)\cap X=
\{ba,b,ba^2,bab,\varepsilon, b^2, ba^3,ba^2b,baba,bab^2\}\cap X=\{ba,b^2\}$.
Classically, $X$ is a maximal code, therefore, since it is $\underline{{\cal P}_1}$-independent,  $X$ is maximal   in the family of ${\underline{\cal P}_1}$-independent codes that is, $X$ satisfies Cond. \ref{3}.
Consequently, since we have $X\subsetneq {\cal P}_1(X)$,  the code $X$ cannot satisfy Cond. \ref{4} (we verify that we have $\varepsilon\in{\cal P}_1(X)$).

\smallskip
(2) For $n\ge k+1$, the  complete uniform code $Y=A^n$ satisfies Conds. \ref{1}, \ref{3}.
Since $Y$ is a maximal code, it  cannot satisfies Cond. \ref{4} (we have $Y\subsetneq {\cal P}_1(Y)$).
 Cond. \ref{2} is no more satisfied by $Y$:
indeed, given two different characters  $a,b$, we have ${\cal P}_k(a^n)\neq \emptyset$ and $a^{n-1}b\in{\cal P}_k^{-1}\left({\cal P}_k\left(a^n\right)\right)\cap Y$
 thus ${\cal P}_k^{-1}\left({\cal P}_k\left(a^n\right)\right)\cap Y\neq \{a^n\}$.

\smallskip
(3) The  regular bifix code $Z=\{ab^na: n\ge 0\}\cup \{ba^nb: n\ge 0\}$ satisfies Cond. \ref{1} wrt. ${\cal P}_1$. 
Indeed, we have:

 ${\underline{\cal P}_1}(Z)=\bigcup_{n\ge 0}\{ab^n,ab^na^2,ab^nab, ba^n, ba^nba,  ba^nb^2\}$,
 thus ${\underline{\cal P}_1}(Z)\cap Z=\emptyset$.\\
 For $n\ne 0$ we have 
 ${\cal P}_1^{-1}{\cal P}_1\left(ab^na\right)=\{ab^n, ab^{n-1}, ab^na,  ab^{n+1}, ab^na^2, ab^na^3,ab^nab,\\ab^na^2b,ab^naba,ab^nab^2
\}$, moreover we have ${\cal P}_1^{-1}{\cal P}_1\left(a^2\right)= {\cal P}_1\left(\{a,a^3,a^2b\}\right)$, thus 
${\cal P}_1^{-1}{\cal P}_1\left(a^2\right)=\{a,  \varepsilon, a^2, ab, a^3, a^4, a^3b,   a^2b,a^2ba,a^2b^2\}$: 
in any case we obtain $Z\cap {\cal P}_1^{-1}{\cal P}_1\left(ab^na\right)=\{ab^na\}$.
 Similarly,  we have $Z\cap {\cal P}_1^{-1}{\cal P}_1\left(ba^nb\right)=\{ba^nb\}$, hence $Z$ satisfies Cond. \ref{2}. At last, we have $\mu(X)=2\cdot1/4\sum_{n\ge 0}(1/2)^n=1$ therefore,
according to Theorem \ref{classic}, $Z$ is a maximal code, whence it is maximal in the family of ${\cal P}_1$-independent codes (Cond. \ref{3}).
Since we have $Z\subsetneq {\cal P}_1(Z)$,  $Z$ cannot satisfies Cond. \ref{4} (we verify that have $a,a^2\in {\cal P}_1(Z)\subseteq \widehat{{\cal P}_1}(Z)$).
\end{example}
In the sequel, we will prove that,  given a regular code $X$, one can decide whether any of Conds. \ref{1}-\ref{4} holds.
Beforehand we establish the following property which, regarding Cond. \ref{1}, plays a prominent part:
\begin{proposition}
\label{Pi-regular}
For every $k\ge 1$, both the relations ${\cal P}_k$ and ${\underline{\cal P}_k}$ are regular.
\end{proposition}
\begin{proof}
In what follows we indicate the construction of a finite automaton with behavior ${\underline{\cal P}_k}$
(see Figure 2). This construction is based on the different underlying configurations in Figure 1.
Firstly, we denote by $E$ the finite set of all the pairs of non-empty words $(u,u')$, with different initial characters, and such that
$|u|+|u'|\le k$.
In addition, $F$ (resp., $G$) stands for the set of all the pairs $(u,\varepsilon)$ (resp., $(\varepsilon,u)$), with $1\le |u|\le k$.
Secondly, we construct a  finite tree-like $A^*\times A^*$-automaton with behavior $E\cup F\cup G$. Let $0$ be the  initial state, the other states being terminal.
We complete the construction by adding the transitions $\left(0,(a,a),0\right)$, for all $a\in A$. Let ${\cal R}_{{\rm P},k}$ be the resulting automaton.\\
By construction we have $\left|{\cal R}_{{\rm P},k}\right|=id_{A^*}\left(E\cup F\cup G\right)$. More precisely,
$\left|{\cal R}_{{\rm P},k}\right|$ is the set of all the pairs $(w,w')$
such that there are  $v,u,u'\in A^*$  satisfying each of  the three following conditions:

(1) $w=pu$, $w'=pu'$; 

(2) if both the words $u,u'$ are non-empty, their initial characters are different; 

(3) $1\le |u|+|u'|\le k$.

\smallskip
\noindent
In other words, $\left|{\cal R}_{{\rm P},k}\right|$ is the sets of   all the pairs $(w,w')$ such that
$1\le d_{\rm P}(w,w')=|u|+|u'|\le k$,
therefore we have $\underline{{\cal P}_k}=\left|{\cal R}_{{\rm P},k}\right|$. Consequently  $\underline{{\cal P}_k}$ and ${\cal P}_k=\underline{{\cal P}_k}\cup id_{A^*}$ are regular relations.
~~~~~~~~~~~~~~~~~~~~~~~~~~~~~~~~~~~~~~~~~~~~~~~~~~~~~~~~~~~~~~~~~~~~~~~~~~~
$\Box$
\end{proof}
\begin{figure}
\begin{center}
\label{Figure1}
\includegraphics[width=5cm,height=5cm]{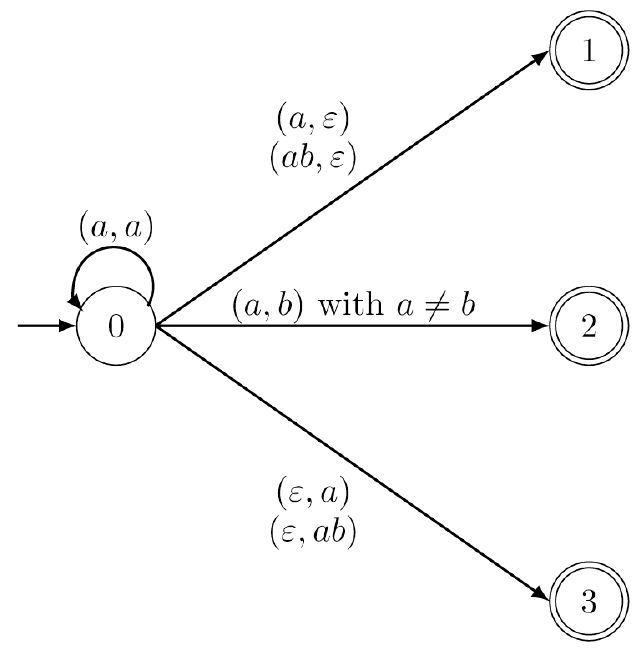}
\end{center}
\caption[]{\small The case where we have $k=2$: in the automaton ${\cal R}_{\rm P, k}$, the   arrows are muti-labelled ($a$, $b$ stand for all pairs of characters in $A$) and terminal states are represented with double circles.}
\end{figure}
\begin{remark}
In \cite{Ng16} the author introduces a peculiar $(A^*\times A^*)\times {\mathbb N}$-automaton: in view of this,  
 for every $(w,w')\in A^*\times A^*$, 
the distance $d_{\rm P}(w,w')$ is the least $d\in{\mathbb N}$ for which $\left((w,w'),d\right)$ is the label of  some successful path.
Furthermore, some alternative proof of the regularity of ${\cal P}_k$ can be obtained.
 However,  we note that such a construction cannot involve the relation $\underline{{\cal P}_k}$ itself that is, it does not affect  Cond. \ref{1}.
\end{remark}
\noindent
The following property is also used in the proof of Proposition  \ref{decid-Pi1}:
\begin{lemma}
\label{Pi2}
Given a positive integer $k$ we have ${\rm \cal P}_{2k}={\rm \cal P}_k^2$.
\end{lemma}
\begin{proof}
In order to prove that ${\rm \cal P}_k^2\subseteq{\rm \cal P}_{2k}$,
we consider a pair of words $(w,w')\in {\rm \cal P}_k^2$. 
By definition, some word $w''\in A^*$ exists such that we have  $(w,w''),(w'',w')\in {\rm \cal P}_k$ that is, 
  $d_{\rm P}(w,w'')\le k$, $d_{\rm P}(w'',w')\le k$. This implies 
$d_{\rm P}(w,w')\le  d_{\rm P}(w,w'')+d_{\rm P}(w'',w')\le 2k$ that is, $ (w,w')\in {\rm \cal P}_{2k}$.\\
Conversely, let  $(w,w')\in {\rm \cal P}_{2k}$, and let $p=w\wedge w'$.
Regarding the integers  $|w|-|p|$, $|w'|-|p|$ , exactly one of the two following conditions occurs:

\smallskip
(a) Firstly, at least one of the integers $|w|-|p|$, $|w'|-|p|$ belongs to $[k+1,2k]$.
Since ${\rm \cal P}_{2k}$ is a symmetric relation, without loss of generality, we assume $k+1\le |w|-|p|\le 2k$. With this condition
a non-empty word $v$ exists such that $w=pvA^k$ that is, $pv\in{\cal P}_k(w)$  with $|pv|=|w|-k$.
On the other hand,
we have $p=pv\wedge w'$, thus $d_{\rm P}(pv,w')=
(|w|-k)+|w'|-2|p|=d_{\rm P}(w,w')-k$. It follows from $d_{\rm P}(w,w')\le 2k$
that $d_{\rm P}(pv,w')\le k$, thus $w'\in{\rm \cal P}_k(pv)$: this implies  $w'\in{\cal P}_k\left({\cal P}_k(w)\right)$, thus $(w,w')\in{\cal P}^2$.

\smallskip
(b) Secondly, in the case where we have   $|w|-|p|\le k$ and $|w'|-|p|\le k$, by definition we have  
$p\in{\cal P}_k(w)$, $w'\in{\cal P}_k(w)$, thus
$(w,w')\in {\rm \cal P}_k^2$.%
~~~~~~~~~~~~~~~~~~~~~~~~~~~~~~~~~~
$\Box$
\end{proof}
We are now ready to establish the following result:
\begin{proposition}
\label{decid-Pi1}
Given  a regular code $X\subseteq A^*$,  wrt. ${\rm\cal P}_k$,
it can be decided whether $X$ satisfies any of Conds. \ref{1}, \ref{2}, and  \ref{4}.
\end{proposition}
\begin{proof} Let $X$ be a regular code. We consider one by one our Conds. \ref{1}, \ref{2}, and \ref{4}:

\smallskip
-- {\it Cond. \ref{1}} According to Proposition \ref{Pi-regular},  $\underline{{\rm \cal P}_k}(X)$ is a regular set, hence $\underline{{\rm \cal P}_k}(X)\cap X$ itself is regular, therefore
one can decide whether Cond. \ref{1} holds.

\smallskip
-- {\it Cond. \ref{2}}  Since ${\cal P}_k$ is a symmetric binary relation, and according to Lemma \ref{Pi2},
we have:
 ${\cal P}_k^{-1}\left({\cal P}_k\left(X\right)\right)\cap X={\cal P}_k^{2}(X)\cap X={\cal P}_{2k}(X)\cap X$.
In addition $x\in X$ implies ${\cal P}_k(x)\ne\emptyset$, therefore Cond. \ref{2}
is equivalent to  $(X\times X)\cap{\cal P}_{2k}\subseteq id_{A^*}$.
This last condition is equivalent to $(X\times X)\cap\left ({\rm\cal P}_{2k}\cap\overline{id_{A^*}}\right)=\emptyset$, thus $
(X\times X)\cap \underline{{\rm\cal P}_{2k}}=\emptyset$.
According to Proposition  \ref{Pi-regular} and 
since $X\times X$ is a recognizable relation,
the set $(X\times X)\cap \underline{{\rm\cal P}_{2k}}$ is regular, therefore one can decide whether  $X$ satisfies Cond. {\ref 2}.

\smallskip
-- {\it Cond. \ref{4}}
According to Proposition \ref{Pi-regular},  the set  ${\cal P}_k(X)$ itself is regular. Consequently,  by applying Sardinas and Patterson algorithm, it can be decided whether $X$ satisfies Cond. \ref{4}.
~~~~~~~~~~~~~~~~~~~~~~~~~~~~~~~~~~~~~~~~~~~~~~~~~~~~~~~~~~~~~~
$\Box$
\end{proof}
It remains to study the bahaviour of $X$ wrt. Cond. \ref{3}.
In order to do this, we proceed by establishing  the two  following results:
\begin{proposition}
\label{Comp-Pref-Independent}
Let $X\subseteq A^*$ be a non-complete regular $\underline{{\rm \cal P}_k}$-independent  code. Then a  complete regular $\underline{{\rm \cal P}_k}$-independent code containing $X$ exists.
\end{proposition}
\begin{proof}
Beforehand, in view of Theorem \ref{EhRz}, we indicate the construction of  a convenient  word $z \in A^*\setminus {\rm F}(X^*)$.
Since $X$ is a non-complete set, by definition some word $z_0$ exists in $A^*\setminus {\rm F}(X^*)$: without loss of generality, we assume $|z_0|\ge k$ (otherwise, we substitute to $z_0$ any word in $z_0A^{k-|z_0|}$).
 Let $a$ be the initial character of $z_0$,
 $b$ be a character different of $a$, and $z=z_0ab^{|z_0|}$.  
Classically,  $z$ is overlapping-free (see e.g.  \cite[Proposition 1.3.6]{BPR10}):
set $U=A^*\setminus (X^*\cup A^*zA^*)$, $Y=z(Uz)^*$, and $Z=X\cup Y$.

\smallskip
According to Theorem \ref{EhRz}, the set $Z$ is a (regular)  complete code. For proving that $Z$ is $\underline{{\rm \cal P}_k}$-independent, we  argue by  contradiction.
By assuming that $\underline{{\rm \cal P}_k}(Z)\cap Z\ne\emptyset$, according to the construction of $Z$, exactly one of the two following cases occurs:

\smallskip
(a) Firstly, $x\in X$, $y\in Y$ exist such that $(x,y)\in \underline{{\rm \cal P}_k}$.
With this condition we have $d_{\rm P}(x,y)=d_{\rm P}(y,x)=\left|(x\wedge y)^{-1}x\right|+\left|(x\wedge y)^{-1}y\right|\le k$.
According to the  construction of $Y$,  the word $z=z_0ab^{|z_0|}$ is a suffix of $y$. It follows from 
$\left|(x\wedge y)^{-1}y\right|\le k\le |z_0|$ that $z_0ab^{|z_0|-1}\in {\rm F}(x\wedge y)$, thus $z_0\in {\rm F}(x)$:
a  contradiction with  $z_0\notin {\rm F}(X^*)$. 

\smallskip
(b) Secondly, $y,y'\in Y$ exist such that $(y,y')\in\underline{ {\rm \cal P}_k}$. Let $p=y\wedge y'$, $u=p^{-1}y$, and  $u'=p^{-1}y'$.

\smallskip
~~~~~~(b1)
At first, assume  $p\in\{y,y'\}$ that is, 
without loss of generality, $y'=p$. With such a condition, $y'$ is a prefix of $y$.
Since we have $y,y'\in z(Uz)^*$, 
and since $z$ is an overlapping-free word, necessarily a word $v\in (Uz)^*$ exists such that $y=y'v$.
 It follows from $|v|=d_{\rm P}(y,y')\le k\le |z|-1$ that  $v=\varepsilon$, thus $y=y'$: 
a contradiction with  $\underline{{\rm\cal P}_k}$ being anti-reflexive.

\smallskip
~~~~~~(b2)
Consequently we have $p\notin\{y,y'\}$ that is, $1\le |u|\le k$ and   $1\le |u'|\le k$ (see Figure 3).
By construction, we have   $b^{|z_0|}\in {\rm S}(z)\subseteq {\rm S}(y)\cap {\rm S}(y')$:
it follows from $1\le |u|\le k\le |z_0|$ and   $1\le |u'|\le k\le |z_0|$ that $u,u'\in {\rm S}(b^{|z_0|})\setminus\{\varepsilon\}$, thus
 $u, u'\in bb^*$:
a contradiction with  $p=y\wedge y'$.

\smallskip
\noindent
In any case we obtain a contradiction, therefore $Z$ is $\underline{{\rm \cal P}_k}$-independent.
~~~~~~~~~~~~~~~~~~
$\Box$
\end{proof}
\begin{figure}
\begin{center}
\label{FigurePrefixProposition3}
\includegraphics[width=8cm,height=6cm]{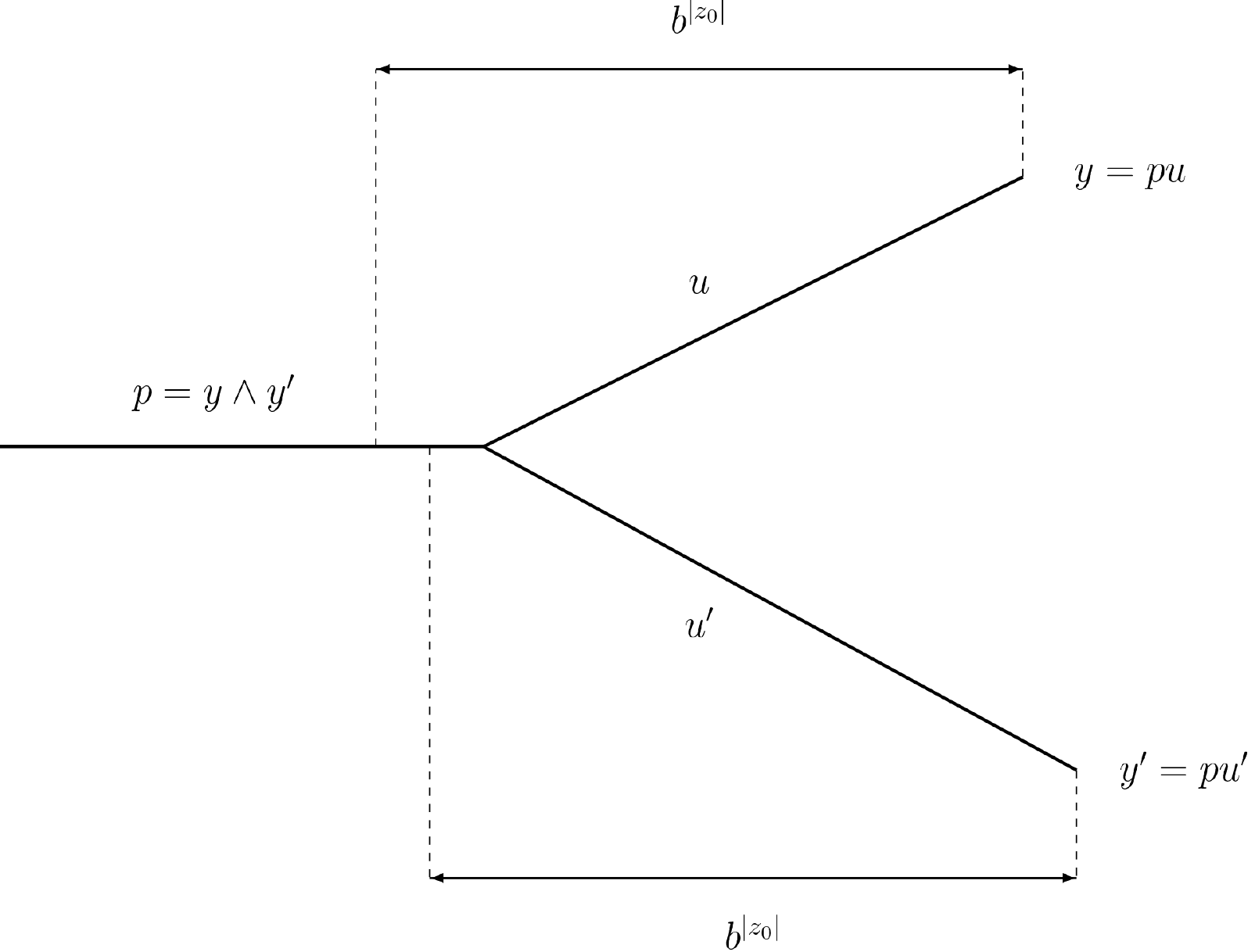}
\end{center}
\caption[]{\small Proof of Proposition \ref{Comp-Pref-Independent}: the case where we have $y,y'\in Y$ and $(y,y')\in{\cal P}_k$, with $p\notin\{y,y'\}$.}
\end{figure}
\begin{proposition}
\label{classic11-Pi-k}
Let $X\subseteq A^*$ be  a regular code.  
Then $X$  is maximal in the family of $\underline{{\rm \cal P}_k}$-independent codes if, and only if, we have $\mu(X)=1$.
\end{proposition}
\begin{proof}
According to Theorem \ref{classic}, $\mu(X)=1$ implies $X$ being a maximal code, thus  $X$ being maximal as a $\underline{{\rm \cal P}_k}$-independent code. For the converse, we argue by contrapositive.
Once more according to Theorem \ref{classic}, $\mu(X)\ne 1$ implies $X$ non-complete.
According to Proposition \ref{Comp-Pref-Independent}, some $\underline{{\rm \cal P}_k}$-independent code strictly containing $X$ exists, hence $X$ is not maximal as a $\underline{{\rm \cal P}_k}$-independent code.
~~~~~~~~~~~~~~~~~~~~~~~~~~~
$\Box$
\end{proof}
\noindent
If $X$ is a regular set, $\mu(X)$ can be computed by making use of some rational series. As a consequence, we obtain the following result:
\begin{proposition}
\label{c3_Pi}
One can decide whether a given regular code $X$ satisfies  Cond. \ref{3}  wrt. ${\rm \cal P}_k$. 
\end{proposition}
\begin{remark}
Given a pair of words $w,w'$, their {\it suffix} distance is $d_{\rm S}= |w|+|w'|-2|s|$,
where  $s$ denotes the longest word in ${\rm S}(w)\cap{\rm S}(w')$. Let ${\cal S}_{k}$ be  the  binary relation
defined by $(w,w')\in {\cal S}_{k}$ if, and only if,  $d_{\rm S}(w,w')\le k$.
Given a word $w\in A^*$, denote by $w^R$ its {\it reversal} that is,  for $a_1,\cdots,a_n\in A$,  we have $w^R=a_n\cdots a_1$ if, and only if,  $w=a_1\cdots a_n$ holds.
For every pair $w,w'\in A^*$, we have $d_{\rm S}(w,w')=d_{\rm P}(w^R,w'^R)$, hence  
$(w,w')\in {\cal S}_k$ is equivalent to $(w^R,w'^R)\in {\cal P}_k$.
As a consequence, given a regular code, one can decide whether it satisfies any of Conds.  \ref{1}--\ref{4} wrt. ${\cal S}_k$. 
\end{remark}
%
\section{Error detection and the  factor metric}
\label{Phi}
By definition, given a pair of words $w,w'\in A^*$, at least one tuple of words, say $(u,v,u',v')$, exists such that $d_{\rm F}(w,w')=|u|+|v|+|u'|+|v'|$.
More precisely, we have $w=ufv$, $w'=u'fv'$, with $f$ being of maximum length.
Such a configuration is illustrated by Figure 4 which, in addition, can provide some support in view of examining the proof of Proposition \ref{Com-Fact--Independent}.

Actually the word $f$, thus the tuple $(u,v,u',v')$, needs not to be unique (see Example \ref{Facts}(1)).
Due to this fact, the construction in the proof of the preceding proposition \ref{Pi-regular} unfortunately  cannot be extended
in order to obtain a finite automaton with behaviour $\underline{\cal F}_k$.

For every positive integer $k$, the relation ${\cal F}_k$ is reflexive and symmetric that is, the equality ${\cal F}_k^{-1}={\cal F}_k$ holds.
In addition, with the preceding notation, we have   ${\cal P}_{k}\cup{\cal S}_{k}\subseteq{\cal F}_{k}\subseteq {\cal F}_{k+1}$.
\begin{figure}
\begin{center}
\label{FigureFactorDistance}
\includegraphics[width=10cm,height=4cm]{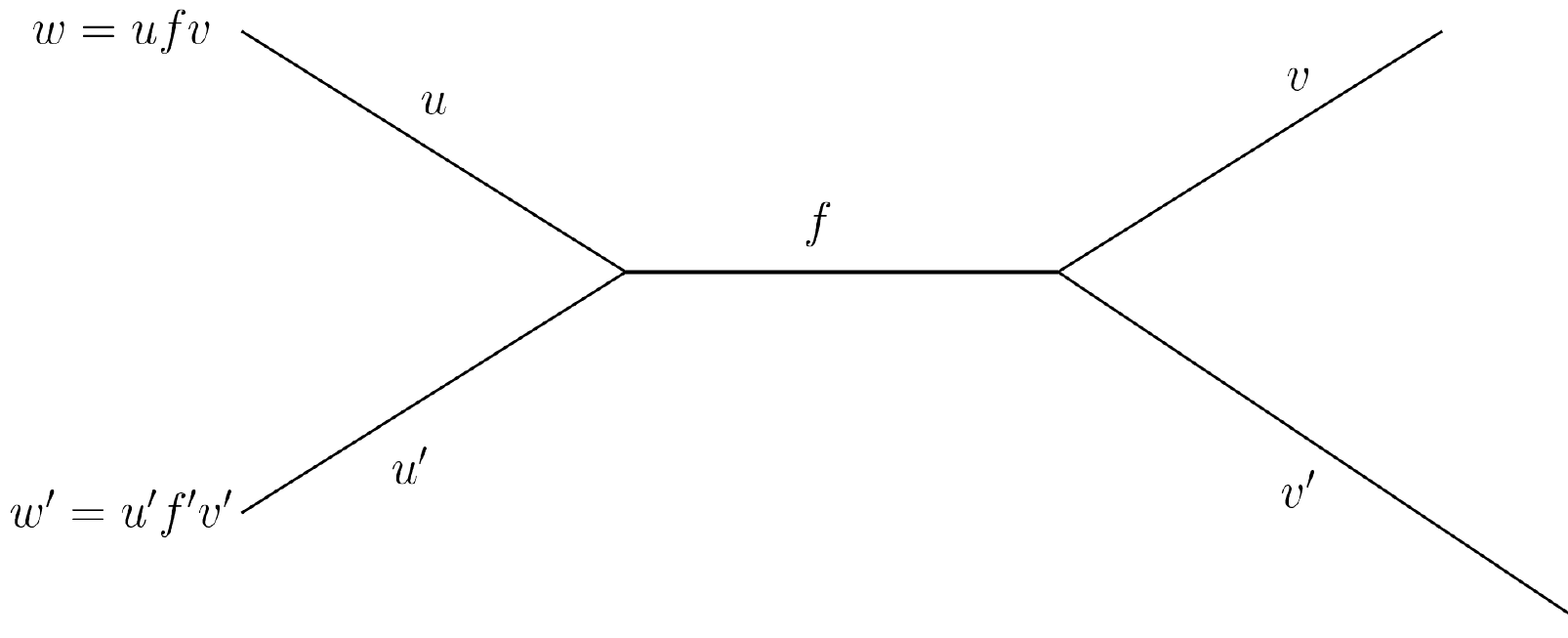}
\end{center}
\caption[]{We have $(w,w')\in{\cal F}_k$ iff.  $|u|+|v|+|u'|+|v'|\le k$.} 
\end{figure}
\begin{example}
\label{Facts}
(1) Let $w=baba babba b$ and $w'=b babba ababaa$. There are two words of maximum length in $F(w)\cap F(w')$, 
namely $f_1=ababa$ and $f_2=babba$. We have $d_{\rm F}(w,w')=|u_1|+|v_1|+|u'_1|+|v'_1|=|u_2|+|v_2|+|u'_2|+|v'_2|=12$,
where the tuples $(u_1,v_1,u'_1,v'_1)$ and $(u_2,v_2,u'_2,v'_2)$ satisfy the following equations:
\begin{eqnarray}
 w=u_1f_1v_1,w'=u'_1f_1v'_1,~~ {\rm with }~~u_1=b,v_1=bbab,u'_1=bbabba,v'_1=a,\nonumber\\
 w=u_2f_2v_2,w'=u'_2f_2v'_2,~~ {\rm with }~~u_2=baba,v_2=b,u'_2=b,v'_2=ababaa.\nonumber
\end{eqnarray}
\smallbreak
(2) 
Over the alphabet $\{a,b\}$, the code $X=\{a,ba,bb\}$ from Example \ref{Prefs}  cannot satisfy Cond. \ref{1} wrt. ${\cal F}_1$.
 Indeed, we have $ba\in\underline{{\cal F}_1}(a)\cap X$.
Since we have 
 ${\cal F}_1^{-1}{\cal F}_1\left(a\right)\cap X=\{a,ba\}$, $X$ cannot satisfy Cond. \ref{2}.
Althought  $X$ is complete, since it does not belong to the family of $\underline{{\cal F}_1}$-independent codes, $X$ cannot satisfy Cond. \ref{3} wrt. ${\cal F}_1$.
It follow from $\varepsilon\in {\cal P}_1(X)\subseteq {\cal F}_1(X)$ that  $X$ can no more satisfy Cond. \ref{4}.

\smallskip
(3) Take $A=\{a,b\}$ and consider the context-free bifix code $Y=\{a^nb^n: n\ge 1\}$.
We have $\underline{{\cal F}_1}(Y)=\bigcup_{n\ge 1}\{a^{n-1}b^n, a^{n+1}b^n, ba^nb^n, a^nb^{n-1}, a^nb^{n}a, a^nb^{n+1}\}$.
This implies $\underline{{\cal F}_1}(Y)\cap Y=\emptyset$, thus $Y$ being $1$-error-detecting wrt. ${\cal F}_1$ (Cond. \ref{1}). 
Regarding error correction, we have $a^{n+1}b^{n+1}\in{\cal F}_1^2(a^nb^n)$, therefore $X$ cannot satisfy Cond. \ref{2} wrt. ${\cal F}_1$. We have $\mu(X)=\sum_{n\ge 1} \left(\frac{1}{4}\right)^n<1$, whence $Y$  cannot satisfy Cond. \ref{3}.
Finally, since we have $(a^nb^{n-1})(ba^nb^n)=(a^nb^n)(a^nb^n)$, the set $\widehat{{\cal F}_1}(Y)={\cal F}_1(Y)$  cannot satisfy Cond. \ref{4}.
\end{example}

\noindent
The following property allows some noticeable connection between the frameworks of prefix, suffix, and factor metrics:
\begin{lemma}
\label{Phi-2}
Given a positive integer $k$ we have ${\cal F}_{k}={\ \cal F}_{1}^k=({\cal P}_1\cup{\cal S}_1)^k$.
\end{lemma}
\begin{proof}
-- We start by proving that we have ${\cal F}_{1}^k=({\cal P}_1\cup{\cal S}_1)^k$.
A indicated above, we have ${\cal P}_1\cup{\cal S}_1\subseteq {\cal F}_{1}$.
Conversely,  given $(w,w')\in {\cal F}_{1}$,  
some tuple of words $(u,v,u',v')$ exists such that $w=ufv$, $w'=u'fv'$, with $0\le |u|+|v|+|u'|+|v'|\le 1$, thus $|u|+|v|+|u'|+|v'|\in\{0,1\}$.
More precisely, at most one element of the set $\{|u|,|v|,|u'|,|v'|\}$ is a non-zero integer: this implies $(w,w')\in{\cal P}_1\cup{\cal S}_1$.
Consequently we have ${\cal F}_{1}={\cal P}_1\cup{\cal S}_1$, thus ${\cal F}_{1}^k=({\cal P}_1\cup{\cal S}_1)^k$.

-- Now, we prove we have ${\cal F}_{1}^k\subseteq {\cal F}_{k} $. Given a pair of words $(w,w')\in {\rm \cal F}_{1}^k$,
 there is some  sequence of words $(w_i)_{0\le i\le k}$  such that $w=w_0$, $w'=w_{k}$, and $d_{\rm F}(w_i,w_{i+1})\le 1$, for each $i\in [0, k-1]$.
We have $d_{\rm F}(w,w')\le \sum_{0\le i\le k-1}d_{\rm F}(w_i,w_{i+1})\le k$, thus $(w,w')\in  {\cal F}_{k}$.

-- For proving that  the inclusion ${\cal F}_{k}\subseteq {\cal F}_{1}^k$ holds,
we argue by induction over $k\ge 1$. The property trivially holds for $k=1$.
Assume that we have ${\rm \cal F}_{k}\subseteq {\rm \cal F}_1^k$, for some $k\ge 1$.  Let $(w,w')\in {\rm \cal F}_{k+1}$ and 
let $f\in {\rm F}(w)\cap {\rm F}(w')$ be a word with maximum length; set $w=ufv$, $w'=u'fv'$. 

\smallskip
(a) Firstly, assume that at least one of the integers  $|w|-|f|$, $|w'|-|f|$ belongs to $[2,k+1]$ that is, without loss of generality  $2\le |w|-|f|=|u|+|v|\le k+1$.
With this condition, there are words $s\in {\rm S}(u)$, $p\in {\rm P}(v)$ such  that $w\in A^hsfpA^{h'}$, with $sp\ne\varepsilon$  and  $h+h'=1$.
On a first hand, it follows from $sfp\in {\rm F}(w)$ that  $d_{\rm F}(w,sfp)=|w|-|sfp|=h+h'=1$, thus $(w,sfp)\in{\rm \cal F}_1$.
On the other hand, $f$ remains  a word of maximum length in ${\rm F}(sfp)\cap {\rm F}(w')$ (otherwise words $s'\in{\rm S}(s)$,  $p'\in {\rm P}(p)$ exist such that we have $s'fp'\in {\rm F}(w)\cap {\rm F}(w')$, with  $|s'|+|p'|\ge 1$), whence we have 
$d_{\rm F}(sfp,w')=|sfp|+|w'|-2|f|$. Since we have $|w|-|sfp|=1$, we obtain $d_{\rm F}(sfp,w') =(|w|-1)+|w'|-2|f|=d_{\rm F}(w,w')-1\le (k+1)-1$,
thus $(sfp,w')\in{\cal F}_k$ that is, by induction, $(sfp,w')\in {\cal F}_1^k$. Since we have  $(w,sfp)\in{\rm \cal F}_1$, this implies $(w,w')\in {\cal F}_1^{k+1}$.

\smallskip
(b) In the case where we have  $|w|-|f|\le 1$ and $|w'|-|f|\le 1$ that is,  $(w,f)\in {\cal F}_1$ and $(f,w')\in  {\cal F}_k$, 
by induction we obtain $(f,w')\in {\cal F}_1^k$, therefore we have  $(w,w')\in {\cal F}_1^{k+1}$.

\smallskip
\noindent 
In any case the condition $(w,w')\in {\rm \cal F}_{k+1}$ implies  $(w,w')\in {\cal F}_1^{k+1}$, hence we have ${\rm \cal F}_{k+1}\subseteq {\cal F}_1^{k+1}$.
As a consequence, for every $k\ge 1$ the inclusion  ${\rm \cal F}_{k}\subseteq{\cal F}_1^{k}$ holds: this  completes the proof.
~~~~~~~~~~~~~~~~~~~~~~~~~~~~~~~~~~~~~~~~~~~~~~~~~~~~~~~~~~~~~~~~~~~~~~~~~~~~~
$\Box$
\end{proof}
As a direct consequence of Lemma \ref{Phi-2}, we obtain the following result:
\begin{proposition}
\label{Phi-regular} Each of the following properties holds:

{\rm (i)}
The relation ${\cal F}_k$ is regular.

{\rm (ii)} Given  a regular code $X$, it can be decided whether $X$  satisfies Cond. \ref{4} wrt.  ${\cal F}_k$.

{\rm (iii)} Given  a finite code $X$, one can decide whether $X$ satisfies any of  Conds.  \ref{1}--\ref{4} wrt.  ${\cal F}_k$.
\end{proposition}
\begin{proof}
In view of Sect. \ref{Prefix-metric}, the relations ${\cal P}_1$ and ${\cal S}_1$ are regular, therefore Property (i) comes from Lemma \ref{Phi-2}.
The proof of Property  (ii)  is  done by merely substituting ${\cal F}_k$ to ${\cal P}_k$ in the proof of the preceding proposition  \ref{decid-Pi1}.
In the case where  $X$ is finite, the same holds for ${\cal F}_k(X)$ and $\underline{{\cal F}_k}(X)$, furthermore Property (iii) holds.
$\Box$
\end{proof}
For non-finite regular sets, the question of the decidability of Conds. \ref{1}, \ref{2} remains open.
Indeed, presently no  finite automaton with behavior ${\underline{\cal F}_k}={\cal F}_k\cap \overline{id_{A^*}}$ is known. 
Actually, 
the following property holds:
\begin{proposition}
\label{Phi-not-recogn}
For every $k\ge 1$, 
${\cal F}_k$ is a  non-recognizable regular relation. 
\end{proposition}
\begin{proof}
According to Proposition \ref{Phi-regular}, the relation ${\cal F}_k$ is regular.
By contradiction, assume  ${\cal F}_k$ recognizable. As indicated in the preliminaries, with this condition 
a finite set $I$ exists such that  ${\cal F}_k=\bigcup_{i\in I}(T_i\times U_i)$. 
For every $n\ge 0$, we have $(a^nb,a^n)\in{\cal F}_k$ therefore,  since $I$ is finite, there are  $i\in I$ and  $m,n\ge 1$, with $m-n\ge k$, 
such that $(a^nb,a^n),(a^mb,a^m)\in T_i\times U_i$.
This implies $(a^nb,a^m)\in  T_i\times U_i\subseteq{\cal F}_k$, thus $d_{\rm F}(a^nb,a^m)\le k$:
a contradiction with $d_{\rm F}(a^nb,a^m)=m-n+1\ge k+1$.
~~~~~~~~~~~~~~~~~~~~~~~~~~~~~~~~~~~~~~~~~~~~~~~~~~~~~~~~~~~~~~~~~~
~~~~~~~~~~~~~~~~~~~~~~~~~~
$\Box$
\end{proof}
Regarding Cond. \ref{3},  the following result holds:
\begin{proposition}
\label{Com-Fact--Independent}
Every non-complete regular $\underline{{\rm \cal F}_k}$-independent code  can be embedded into some  complete one.
\end{proposition}
\begin{proof}
In the family of $\underline{\rm{\cal F}_k}$-independent codes, we consider a non-complete regular set $X$.
In view of Theorem \ref{EhRz},  we will construct a convenient word $z_1\in A^*\setminus F(X^*)$.
Take a word   $z_0\notin {\rm F}(X^*)$,  with $|z_0|\ge k$; let $a$ be  its initial character, and $b$ be a character different of $a$.
Consider the word $z$ that was constructed in the proof of Proposition \ref{Comp-Pref-Independent}, that is $z=z_0ab^{|z_0|}$.
Set $z_1=a^{|z|}bz=a^{2|z_0|+1}bz_0ab^{|z_0|}$: since by construction, $z_1^R$, the  reversal of $z_1$, is  overlapping-free, the same holds for $z_1$. 
 Set  $U_1=A^*\setminus \left(X^*\cup A^*z_1A^*\right)$, $Y_1=z_1\left(U_1z_1\right)^*$, and $Z_1=X\cup Y_1$. 

\smallskip
According to  Theorem \ref{EhRz}, the set  $Z_1$ is  a regular complete  code. In order to
prove that it   is $\underline{{\rm \cal F}_k}$-independent that is, $\underline{{\rm \cal F}_k}(Z)\cap Z=\emptyset$, we argue by contradiction.
Actually, exactly one of the following conditions occurs:

\smallskip
(a) The first condition sets that $x\in X$, $y\in Y_1$ exist such that  $(x,y)\in\underline{{\rm \cal F}_k}$.
  Let $f$ be a word with maximum length in ${\rm F}(x)\cap {\rm F}(y)$: we have $y=ufv$, with $|u|+|v|\le d_{\rm F}(x,y)\le k$.
It follows from $y\in z_1\left(U_1z_1\right)^*$ and $|z_1|\ge k+1$  that we have  $u\in {\rm P}(z_1)$ and $v\in{\rm S}(z_1)$.
More precisely, according to the construction of $y$,  we have $u\in {\rm P}\left(a^{|z_0|}\right)$,  $v\in{\rm S}\left(b^{|z_0|}\right)$, and $u^{-1}y\in a^{|z_0|+1}A^*$, thus
$f\in A^*z_0A^*$. Since we have $f\in{\rm F}(x)$, we obtain  $z_0\in  {\rm F}(x)$, a contradiction with $z_0\notin {\rm F}(X^*)$.

\smallskip
(b) With the second condition, a pair of different words $y,y'\in Y_1$ exist such that $(y,y')\in{\rm \cal F}_k$. Let $f$ be  a word with maximum length in ${\rm F}(y)\cap {\rm F}(y')$. 
As indicated above, words  $u,u',v,v'$ exist such that $w=ufv$, $w'=u'fv'$, with $|u|+|u'|+|v|+|v'|=d_{\rm F}(w,w')\le k$.

\smallskip
~~~~~~(b1) At first,  assume that both the words $v,v'$ are different of $\varepsilon$. 
According to the construction of $Y_1$, since we have $v,v'\in {\rm S}(Y_1)$ with  $|v|+|v'|\le k$, a pair of positive integers $i,j$ exist such that $v=b^i$, $v'=b^j$.
This implies $fb\in {\rm F}(y)\cap {\rm F}(y')$, which contradicts the maximality of $|f|$.

\smallskip
~~~~~~(b2) As a consequence, at least one of the conditions $v=\varepsilon$, or $v'=\varepsilon$ holds. Without loss of generality, we assume $v'=\varepsilon$,
thus $f\in{\rm  S}(y')$. On a first hand, it follows from $z_1\in {\rm F}(y)\cap {\rm F}(y')$ that $|f|\ge |z_1|$:
since we have $f,z_1\in {\rm S}(y')$, this implies $f\in A^*z_1$. 
On the other hand, since we have  $z_1\in {\rm S}(y)$, $fv\in {\rm S}(y)$, and $|f|\ge |z_1|$, we obtain $z_1\in {\rm S}(fv)$.
This implies $fv\in A^*z_1v\cap A^*z_1$, thus $z_1v\in A^*z_1$. It follows from $|v|\le |z_1|-1$ and $z_1$ being  overlapping-free that $v=\varepsilon$.
Similar arguments applied to the prefixes of $y,y'$ lead to $u=u'=\varepsilon$:
we obtain $y=y'$, a contradiction with $(y,y')\in \underline{{\cal F}_k}$.

\smallskip
\noindent
In any case we obtain a contradiction, therefore $Z_1$ is $\underline{{\rm \cal F}_k}$-independent.
~~~~~~~~~~~~~~~~~~
$\Box$
\end{proof}
As a consequence, by merely substituting ${\rm\cal F}_k$ to ${\rm\cal P}_k$ in the proof of the propositions \ref{classic11-Pi-k} and \ref{c3_Pi},  we obtain the following statement:
\begin{proposition}
\label{classic-22-Phi-k}
Given a  regular code $X$, each of the  following properties holds:

{\rm (i)} $X$ is maximal as a $\underline{{\rm \cal F}_k}$-independent code if, and only if, we have $\mu(X)=1$.

{\rm (ii)} One  can decide whether $X$ satisfies Cond. \ref{3} wrt. ${\rm \cal F}_k$.
\end{proposition}
%
\section{Error detection in the topologies associated to (anti-)automorphisms}
\label{antiautomorphism}
Given an (anti-)automorphism $\theta$, we will examine Conds. \ref{1}--\ref{4} wrt. the quasi-metric $d_\theta$ that is, wrt. the relation $\tau_{d_\theta,1}=\widehat\theta$.
Recall that we have $\underline {\left(\widehat\theta\right)}=\underline\theta$. Regarding error correction, the following noticeable property holds:
\begin{proposition}
\label{1-2-equiv}
With respect to $\widehat\theta$, a  regular code $X\subseteq A^*$ satisfies Cond. \ref{1}  if, and only if, it satisfies Cond. \ref{2}. 
\end{proposition}
\begin{proof}
-- Firstly, assume  that $X$ is  $\underline\theta$-independent, and let $x,y\in X$ such that $\underline{\tau_{d_\theta,1}}(x)\cap \underline{\tau_{d_\theta,1}}(y)=\underline\theta(x)\cap\underline\theta(y)\ne\emptyset$. It follows from $\underline\theta(x)\cap\underline\theta(y)=\left(\{\theta(x)\}\setminus\{x\}\right)\cap\left(\{\theta(y)\}\setminus\{y\}\right)$ that $\theta(x)\ne x$, $\theta(y)\ne y$, and $\theta(x)=\theta(y)$: since $\theta$ is one-to-one, this implies $x=y$, therefore $X$ satisfies Cond. \ref{2}.

-- Secondly,  assume that Cond. \ref{1} does not hold wrt. $\hat\theta$ that is, $X\cap\underline{(\hat\theta)}(X)=X\cap\underline\theta(X)\ne\emptyset$.
Necessarily, a pair of different words $x,y\in X$ exist such that $y=\theta(x)$.
It follows from $\widehat\theta(x)=\{x\}\cup\{\theta(x)\}=\{x,y\}$ and  $\widehat\theta(y)=\{y\}\cup\{\theta(y)\}$ that $\widehat\theta(x)\cap \widehat\theta(y)\ne\emptyset$, hence 
Cond.  \ref{2} cannot hold.
~~~~~~~~~~~~~~~~~~~~~~~~~~~~~~~
$\Box$
\end{proof}
\begin{example}
\label{theta}
(1) Let  $A=\{a,b\}$ and $\theta$ be the automorphism defined by $\theta(a)=b$, and $\theta(b)=a$.
The regular prefix code $X=\{a^nb: n\ge 0\}$ satisfies Conds. \ref{1}. Indeed, we have $\underline\theta(X)=\{b^na: n\ge 0\}$, thus $\underline\theta(X)\cap X=\emptyset$.
According to Proposition \ref{1-2-equiv}, it also satisfies Cond.  \ref{2}.
We have $\mu(X)=\frac{1}{2}\sum_{n\ge 0}\left(\frac{1}{2}\right)^n=1$, whence $X$ is a maximal prefix code.
Consequently $X$ is maximal in the family of $\underline\theta$-independent codes (Cond. \ref{3}).  Finally, we have $X\subsetneq\hat\theta(X)=\{a^nb: n\ge 0\}\cup \{b^na: n\ge0\}$, hence $\hat\theta(X)$ cannot be  a code (we verify that $a,b,ab\in \hat\theta(X)$).

\smallbreak
(2) Over the alphabet  $A=\{a,b\}$,  take for $\theta$ the anti-automorphism defined by $\theta(a)=b$, and $\theta(b)=a$.
The regular prefix code $X=\{a^nb: n\ge 1\}$ satisfies Conds. \ref{1}, \ref{2}: indeed, we have $\underline\theta(X)=\{ab^n: n\ne 1\}$, thus $\underline\theta(X)\cap X=\emptyset$.
As indicated above, $X$ is a maximal prefix code, thus it satisfies Cond.\ref{3}.
At last, it follows from $X\subsetneq \widehat\theta(X)$ that  
$X$ cannot satisfies Cond. \ref{4}. 

\smallbreak
(3) With the condition above,
 consider $Y=X\setminus \{b,ab\}=\{a^nb: n\ge 2\}$. We have $\underline\theta(Y)\cap Y=\emptyset$, hence $Y$ satisfies Conds. \ref{1}, \ref{2}. However, by construction, $Y$ cannot satisfy Cond. \ref{3}.
Finally, we have $\hat\theta(Y)=\bigcup_{n\ge 2}\{a^nb, ab^n\}$, which remains a prefix code that is, Y satisfies Cond. \ref{4}.

\smallbreak
 (4) Over the alphabet $\{A,C,G,T\}$, let  $\theta$ denotes  the Watson-Crick anti-automorphism (see eg. \cite{KKK14,K83}),
which is defined by  $\theta(A)=T$, $\theta(T)=A, \theta(C)=G$, and $\theta(G)=C$.\\
Consider the prefix code $Z=\{A,C,GA,G^2,GT,GCA,GC^2,GCG,GCT\}$.
We have $\theta(Z)=\{T,G,TC,C^2,AC,TGC,G^2C, CGC,AGC\}$, hence $Z$ satisfies Conds. \ref{1}, \ref{2}.
By making use of the uniform distribution, we have $\mu(Z)=1/2+ 3/16+1/16=3/4$, hence $Z$ cannot satisfy Cond. \ref{3}.
At last,  it follows from $G,G^2\in\widehat\theta (Z)=\theta(Z)\cup Z$ that  Cond. \ref{4} is not  satisfied.

\smallbreak
(5) Notice that, in each of the preceding examples, since the (anti-)automorphism $\theta$ satisfies $\theta^2=id_{A^*}$, actually the quasimetric $d_\theta$ is a metric.\\
Of course, (anti-)automorphisms exist in such a way that $d_\theta$ is only a quasi-metric.
For instance over $A=\{a,b,c\}$, taking for $\theta$  the automorphism generated by the cycle $(a,b,c)$, we obtain $d_\theta(a,b)=1$ and $d_\theta(b,a)=2$ 
(we have $b=\theta(a)$ and $a\ne\theta(b)$).
\end{example}
\begin{figure}
\begin{center}
\label{Figure2}
\includegraphics[width=3cm,height=1.9cm]{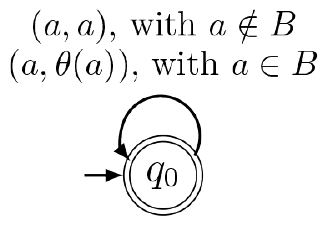}
\end{center}
\caption[]{\small An automaton with behavior $\tau_{d_{\theta,1}}=\hat\theta$, in the case where  $\theta$ is an automorphism:  $a$ represents every character in $A$ and we set $B=\{a\in A: \theta(a)\ne a\}$.}
\end{figure}
Regarding regulary and recognizability of relations we state:
\begin{proposition}
\label{theta-X-regular}
With the preceding  notation, each of the following properties holds:

(i) If $\theta$  is an  automorphism, then $\tau_{d_{\theta,1}}=\hat\theta$ and $\underline{\tau_{d_{\theta,1}}}=\underline\theta$  are non-recognizable regular relations.

(ii) If $\theta$ is an anti-automorphism, then it cannot be a regular relation.

(iii) Given an (anti-)automorphism $\theta$, if $X$ is a regular subset of $A^*$, then the same holds for $\hat\theta(X)$.
\end{proposition}
\begin{proof} Let $\theta$ be an (anti-)automorphism onto $A^*$.

\smallskip
 (i) In the case where  $\theta$  is an  automorphism of $A^*$, it is regular: indeed, trivially $\theta$ is the behavior of a one-state automaton with  transitions $\left(0,\left(a,\theta(a)\right),0\right)$, for all $a\in A$. 
Starting with this automaton we obtain a finite automaton with behaviour $\hat\theta$  by merely adding  the transitions $\left(0,\left(a,a\right),0\right)$, for all $a\in A$ (see  Figure 6): consequently the relation $\hat\theta$ is regular.

By contradiction, assume $\theta$ recognizable. As in the proof of Proposition \ref{Phi-not-recogn}, a finite set $I$ exists such that  $\hat\theta=\bigcup_{i\in I}(T_i\times U_i)$. 
Since $I$ is finite, there are  $i\in I$, $a\in A$, and $m,n\ge 1$  such that $\left(a^n,\left(\hat\theta(a)\right)^n\right),\left(a^m,\left(\hat\theta(a)\right)^m\right)\in T_i\times U_i$, with $m\ne n$.
This implies $\left(a^n,\left(\hat\theta(a)\right)^m\right)\in T_i\times U_i$ 
that is,  $\hat\theta(a^n)=(\hat\theta(a))^m$. If we have $\hat\theta (a)=a$, we obtain $a^n=a^m$. Otherwise, we have $\hat\theta(a)=\theta(a)$, thus $\theta(a^n)=(\theta(a))^m$. 
In each case, this contradicts $m\ne n$.

Finally, the binary relation $\underline{(\hat\theta)}=\underline\theta$ is the set of all the pairs $(uas,ubs')$ that satisfy both  the three following conditions: 

(1) $u\in A^*$;

(2)  $b=\theta(a)$, with $a,b\in A$ and $a\ne b$;

(3) $s'=\theta(s)$, with $s\in A^*$.

\smallbreak
Consequently, $\underline\theta$ is the  behaviour of the automaton in Figure 6, hence it is a regular relation.
In addition, by merely substituting the relation $\underline\theta$ to $\hat\theta$ in the argument above, it can be easily prove that $\underline\theta$ recognizable implies $\theta(a^n)=(\theta(a))^m$, for some $a\in A$ and $m\ne n$: a contradiction with $\theta$ being a free monoid automorphism.

\smallbreak
(ii) For every anti-automorphism $\theta$, the relation $\hat\theta$ is the result of the composition of the so-called {\it transposition},
 namely $t:w\rightarrow w^R$,  by some automorphism of $A^*$, say $h$.
As shown in \cite [Example IV.1.10]{S03}, the transposition is not a  regular relation.
Actually, the same argument can be applied  for proving that  the resulting relation $\hat\theta$ is non-regular.

\smallbreak
(iii) Let $X$ be a regular subset of $A^*$. If $\theta$ is an automorphism, the relation $\hat\theta=\theta\cup id_{A^*}$ is a regular relation, hence $\hat\theta(X)$ is regular.
In the case where $\theta$ is an anti-automorphism, with the preceding notation, although the transposition is not a regular relation, the set $t(X)$ itself is regular  (see eg. \cite[Proposition I.1.1] {S03}). Consequently $\hat\theta(X)=h\left(t(X)\right)$ is also regular. 
~~~~~~~~~~~~~~~~~~~~~~~~~~~~~~~~~~~~$\Box$
\end{proof}
\begin{figure}
\begin{center}
\includegraphics[width=6.5cm,height=1.6cm]{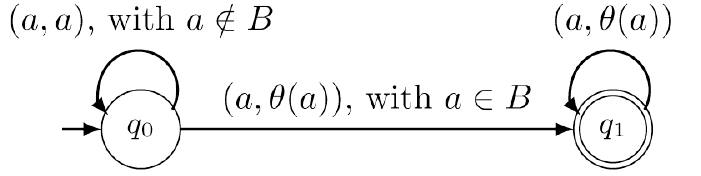}
\end{center}
\caption[]{\small With the notation in Example 6, an automaton with behavior $\underline\theta$, in the case where $\theta$ is an automorphism.}
\end{figure}
As  a consequence of Proposition \ref{theta-X-regular}, we obtain the following result:
\begin{proposition}
Given an (anti-)automorphism $\theta$ onto $A^*$ and a regular code $X\subseteq A^*$, each of the two following properties holds:

(i) In any case, $X$ satisfies Cond. \ref{1}, \ref{2}.

(ii) It can be decided whether $X$ satisfies Cond. \ref{4}.
\end{proposition}
\begin{proof}
(i) Firstly, since $\theta$ is an (anti-)automorphism of $A^*$, it is a  one-to-one mapping.
For every $x\in X$, we have $\theta(x)\ne\emptyset$, moreover we have $\theta^{-1}\left(\theta(x)\right)\cap X=\{x\}$, consequently, in any case $X$ satisfies Cond \ref{2}.
Secondly, according to Proposition \ref{1-2-equiv}, necessarily $X$ also satisfies Cond. \ref{1}.

\smallbreak
\noindent
(ii) According to Proposition \ref{theta-X-regular}, in any case the set $\theta(X)$  is regular:  by applying Sardinas and Patterson algorithm, one can  decide whether the code $X$ satisfies Cond. \ref{4}.
~~~~~~~~~~~~~~~~~~~~~~~~~~~~~~~~~~~~~~~~~~~~~~~~~~~~~~~~~~~~~~~~~~~~~~~~~~~~~~~~~~~~~~~~~~~~~$\Box$
\end{proof}
It remains to study the behavior of regular codes with regard to Cond. \ref{3}:
\begin{proposition}
\label{Embed-theta}
Every non-complete regular $\underline\theta$-independent code can be embedded into some complete one.
\end{proposition}
\begin{proof}
According to Theorem \ref{EhRz}, the result holds if $\theta$ is an automorphism: indeed the action of such a transformation merely consists of  rewriting words by applying some permutation of $A$.

Now, we assume that $\theta$ is  an anti-automorphism. Classically, some positive integer $n$, the {\it order} of the permutation $\theta$,  exists such $\theta^{n}=id_{A^*}$. 
As in the propositions \ref{Comp-Pref-Independent} and \ref{Com-Fact--Independent}, in view of Theorem \ref{EhRz}, we construct a convenient word in $A^*\setminus {\rm F}(X^*)$.

Let $z_0\notin {\rm F}(X^*)$, $a$ be its initial character, and $b$ be a character different of $a$. Without loss of generality,  we assume $|z_0|\ge 2$ and $z_0\notin aa^*$, for every $a\in A$ (otherwise,  substitute $z_0b$ to $z_0$).
By definition, for every integer $i$, we have $\left|\theta^i(z_0)\right|=\left|\theta(z_0)\right|$, therefore it follows from $z_0\theta(z_0)\cdots\theta^{n-1}(z_0)\in A^*\setminus F(X^*)$ that
 $z_2=z_0\theta(z_0)\cdots\theta^{n-1}(z_0)ab^{n|z_0|}$ is an overlapping-free word in $A^*\setminus {\rm F}(X^*)$.
Set $U_2=A^*\setminus \left(X^*\cup A^*z_2A^*\right)$, $Y_2=(z_2U_2)^*z_2$, and $Z_2=X\cup Y_2$.

\smallskip
According to  Theorem \ref{EhRz}, the set $Z_2=X\cup Y_2$ is a complete regular code. 
In order to prove that  it is $\underline\theta$-independent, we argue by contradiction. Actually, assuming that $\underline\theta(Z)\cap Z\ne \emptyset$, exactly one of the three following cases occurs:

\smallskip
(a)  The first condition sets that $x\in X$ exists such that $\underline\theta (x)\in Y_2$.
 By construction, the sets $X$ and $Y_2$ are disjoint (we have $z_2\not\in F(X^*)$), therefore we have $\theta(x)\ne x$: this implies $\underline\theta (x)=\theta(x)$.
According to the definition of $Y_2$, we have $z_2\in {\rm F}\left(\theta(x)\right)$, thus  $\theta(z_0)\in {\rm F}\left(\theta(x)\right)$. 
It follows from $x=\theta^{n-1}\left(\theta(x)\right)$,  that the word $z_0=\theta^{n-1}\left( \theta(z_0\right))$ is a factor of $x$, a contradiction with $z_0\notin {\rm F}(X^*)$.

\smallskip
(b) With the second condition some  pair of words $x\in X$, $y\in Y_2$ exist such that 
$x=\theta(y)$.
It follows from $\theta^{n-1}(z_0)\in {\rm F}(z_2)\subseteq {\rm F}(y)$  that $z_0=\theta\left(\theta^{n-1}(z_0)\right)\in {\rm F}\left(\theta(y)\right)={\rm F}\left(x\right)$: once more this contradicts
$z_0\notin {\rm F}(X^*)$.
\smallskip

(c)  The third condition sets that there are  different words $y, y'\in Y_2$ such that $y'=\theta(y)$. Since $\theta$ is an anti-automorphism,  $ab^{n|z_0|}\in {\rm S}(y)$ implies $\theta\left(b^{n|z_0|}\right)\in {\rm P}(y')$.
Since we have $z_0\in {\rm P}(Y_2)$, this implies
$\left(\theta\left(b\right)\right)^{|z_0|}=z_0$. But we have  $|z_0|\ge 2$, and $\theta(b)\in A$: this is incompatible with the construction of $z_0\not\in aa^*$.

\smallskip
\noindent
In any case we obtain a contradiction, therefore $Z_1$ is $\underline{\theta}$-independent.
~~~~~~~~~~~~~~~~~~
$\Box$
\end{proof}
As a consequence,  given a regular code $X\subseteq A^*$, 
$X$ is maximal in the family of $\underline\theta$-independent  codes of $A^*$ if, and only if, the equation $\mu(X)=1$ holds.
In other words,  
one can decide whether  $X$ satisfies Cond. \ref{3}.

Finally, the following statement synthesizes the results of the whole study we relate in our paper:

\begin{theorem}
\label{decidable}
With the preceding notation, given a regular code $X$
 one can decide whether  it satisfies each of the following conditions:

{\rm (i)}  Conds.  \ref{1}--\ref{4} wrt.   ${\cal P}_k$.

{\rm (ii)}   Conds.  \ref{3}, \ref{4} wrt.  ${\cal F}_k$.

{\rm (iii)} In the case where $X$ is finite, Conds. \ref{1}, \ref{2} wrt.  ${\cal F}_k$.

{\rm (iv)}  Conds. \ref{3}, \ref{4} wrt. $\widehat \theta$, for any (anti-)automorphism  $\theta$  of $A^*$.
In that framework,  in any case $X$ satisfies Conds. (c1), (c2), which  are  actually equivalent.

\end{theorem}
\section{Concluding remark}
\label{conclusion}
With  each pair of words $(w,w')$, the  so-called {\it subsequence metric} associates the integer $\delta(w,w')=|w|+|w'|-2s(w,w')$, where $s(w,w')$ stands for a maximum length common subsequence of $w$ and $w'$.
Equivalently, $\delta(w,w')$ is the minimum number of  one character insertions and  deletions that have to be applied   for computing $w'$ by starting from $w$. 
We observe that, wrt. relation $\tau_{\delta,k}$, results very similar to the ones of the propositions   \ref{Phi-regular} and  \ref{classic-22-Phi-k} hold \cite{N21}.
Moreover, in that framework,  we still do not  know whether  Conds. \ref{1}, \ref{2} can be decided, given a regular code $X$. 

It is noticeable that, although the inclusion ${\cal F}_k\subseteq \tau_{\delta,k}$ holds  (indeed factors are very special subsequences of words), we do not know any more 
whether the relation $\underline{{\cal F}_k}$ is   regular  or not.
In the case where the answer is no,   can after all Conds. \ref{1}, \ref{2} be decidable? 
%

 From another point of view, wrt.   each of the  relations we mentionned, presenting families of codes satisfying all the best Conds. \ref{1}--\ref{2} would be desirable.
\bibliography{mysmallbib}{}
\bibliographystyle{splncs04}
\end{document}